\documentclass[runningheads]{llncs}
\usepackage[T1]{fontenc}
\usepackage{graphicx}
\usepackage{amssymb}            
\usepackage{mathtools}          
\usepackage{mathrsfs}           
\numberwithin{equation}{section} 
\usepackage{graphicx}           
\usepackage{subcaption}         
\usepackage[space]{grffile}     
\usepackage{url}                
\usepackage[backend=bibtex,style=numeric]{biblatex}
\usepackage[hidelinks]{hyperref}
\usepackage{algorithm}
\usepackage[linesnumbered,ruled,vlined,algo2e]{algorithm2e}
\usepackage{svg}
\usepackage{dsfont}
\usepackage[acronym]{glossaries} 
\glsdisablehyper
\usepackage{array}
\usepackage{booktabs}
\usepackage{multirow}
\usepackage{graphicx}
\usepackage[inline]{enumitem}
\setlist{nolistsep,leftmargin=*}
\usepackage{cleveref}
\crefname{figure}{figure}{figures}
\crefname{section}{section}{sections}
\crefname{subsection}{subsection}{subsections}
\crefname{table}{table}{tables}
\usepackage{siunitx}
\usepackage{comment}
\let\LNCSvec\vec   
\let\vec\relax     

\usepackage{accents}

\let\vec\LNCSvec   

\usepackage[normalem]{ulem} 

\newacronym{acr:rvc}{RVC}{request vehicle combination}
\newacronym{acr:hva}{HVA}{Hybrid Vehicle Adjusted}
\newacronym{acr:hv}{HV}{Hybrid Vehicle}
\newacronym{acr:mdp}{MDP}{Markov Decision Process}
\newacronym{acr:rl}{RL}{Reinforcement Learning}
\newacronym{acr:marl}{MARL}{Multi-Agent Reinforcement Learning}
\newacronym{acr:amod}{AMoD}{Autonomous Mobility-on-Demand}
\newacronym{acr:drl}{DRL}{Deep Reinforcement Learning}
\newacronym{acr:mod}{MoD}{Mobility-on-Demand}
\newacronym{acr:sac}{SAC}{Soft Actor-Critic}
\newacronym{acr:sacd}{SACD}{Soft Actor-Critic Discrete}
\newacronym{acr:nn}{NN}{Neural Network}
\newacronym{acr:mpc}{MPC}{Model Predictive Control}
\newacronym{acr:relu}{ReLU}{rectified linear unit}
\begin{document}
\title{Multi-Agent Soft Actor-Critic with Coordinated Loss for Autonomous Mobility-on-Demand Fleet Control}
\titlerunning{Multi-Agent SAC with Coordinated Loss for AMoD Fleet Control}
%
\author{Zeno Woywood\inst{1}* \and
Jasper I. Wiltfang\inst{1}* \and
Julius Luy\inst{1} \and Tobias Enders\inst{1} \and Maximilian Schiffer \inst{1,2}}
\authorrunning{Z. Woywood et al.}
%
\institute{School of Management, Technical University of Munich, \\ Arcisstraße 21, 80333 Munich, Germany, \\ 
\email{\{zeno.woywood, jasper.wiltfang, julius.luy, tobias.enders\}@tum.de} \and Munich Data Science Institute, Technical University of Munich, \\ Walther-von-Dyck-Straße 10, 85748 Garching, Germany \\
\email{schiffer@tum.de}}
\maketitle              
\begin{abstract}
We study a sequential decision-making problem for a profit-maximizing operator of an autonomous mobility-on-demand system. Optimizing a central operator's vehicle-to-request dispatching policy requires efficient and effective fleet control strategies. To this end, we employ a multi-agent Soft Actor-Critic algorithm combined with weighted bipartite matching. We propose a novel vehicle-based algorithm architecture and adapt the critic's loss function to appropriately consider coordinated actions. Furthermore, we extend our algorithm to incorporate rebalancing capabilities. Through numerical experiments, we show that our approach outperforms state-of-the-art benchmarks by up to 12.9\% for dispatching and up to 38.9\% with integrated rebalancing.


\def\thefootnote{*}\footnotetext{Both authors contributed equally to this work.}

\end{abstract}

\section{Introduction}\label{sec:introduction}
\gls{acr:amod} systems promise to transform urban mobility, following \gls{acr:mod} providers like Uber and DiDi. They enable \gls{acr:mod} services with fast response times and point-to-point trips at lower costs, shifting operator priorities: \gls{acr:mod} depends on driver wages, favoring revenue maximization, whereas \gls{acr:amod} focuses on profit optimization by minimizing distance-related costs. Additionally, \gls{acr:amod} enhances central control by leveraging historical and contextual trip data. In this context, operators face two key decisions: accepting and dispatching requests, and rebalancing vehicles to match demand. This results in a stochastic control problem, which we study through the lense of \gls{acr:drl}. Specifically, we introduce a novel parallel algorithm combining multi-agent \gls{acr:drl} with combinatorial optimization for scalability and efficiency.

\textbf{Related Work:}
Control algorithms for (A)\gls{acr:mod} systems range from rule-based heuristics \cite[see, e.g.,][]{Hyland_2018} to (learning-enriched) optimization approaches \cite[see, e.g.,][]{Alonso_Mora2017, Jungel_2023_AMODFleetControl_OnlineOptimization}, and \gls{acr:mpc} \cite[see, e.g.,][]{Iglesias_2018}, but are only loosely related to this work as we focus on \gls{acr:rl}. \gls{acr:rl} adapts online to stochastic demand, whereas \gls{acr:mpc} relies on fixed predictive horizons. \gls{acr:rl} approaches split into single-agent approaches for dispatching \cite{Qin_2020, Liu2022, Zheng_2022} or rebalancing \cite{Fluri, Gammelli_2021_AMODRebalancing_DRLGraphs, Jiao_2021}, and multi-agent approaches for dispatching \cite{Li_2019,Tang_2019, Zhou_2019, Enders_2023_AMODDispatching_HybridMADRL} or rebalancing \cite{Holler_2019_MODRebalancing_MADRL, Li_2022_RebalancingWOCost_MADRL, Liang_2022_CMODRebalancing_MADRL}. Here, multi-agent approaches promise to increase scalability to larger action spaces and remain the focus of our work. Accordingly, we limit our discussion to recent works in this field to ensure conciseness.

For dispatching, \cite{Li_2019} proposed a mean field actor-critic algorithm, while \cite{Tang_2019} used a value-based approach, \cite{Zhou_2019} used Q-Learning with KL-divergence optimization, and \cite{Enders_2023_AMODDispatching_HybridMADRL} employed a \gls{acr:sac} with bipartite matching. 
For combined dispatching and rebalancing, \cite{Li_2022_RebalancingWOCost_MADRL} used a graph neural network with a centralized critic for training, and a decentralized actor-critic for execution, while
\cite{Holler_2019_MODRebalancing_MADRL} used Deep Q-Networks and Proximal Policy Optimization with attention mechanisms to process global states, and 
\cite{Liang_2022_CMODRebalancing_MADRL} proposed a two-stage algorithm for dispatching and rebalancing using centralized programming. All methods treated dispatching and rebalancing as sequential decisions.
Our work contributes to this field by integrating dispatching and rebalancing in a multi-agent \gls{acr:drl} setting. 

\textbf{Contribution:}
To close the research gap outlined above, we present a novel multi-agent actor-critic algorithm for \gls{acr:amod} fleet dispatching and integrated rebalancing under a profit maximization objective. Specifically, we propose a scalable parallel \gls{acr:sac} architecture, wherein each vehicle functions as an agent providing per-request weights, to address the operators' need for a comprehensive and feasible solution to its control problem. We obtain coordinated actions for the fleet by solving a bipartite matching problem, wherein vehicles and requests represent vertices with the agents' per-request weights between them. To account for the differences between per-agent and coordinated actions induced by the combination of \gls{acr:drl} and a coordination layer, we show how to adapt the critic loss function. This ``coordinated loss'' leads to a more precise estimate of the value of future states which is coordinated across the fleet. We show that our algorithm outperforms dispatching-only algorithms by up to 12.9\% on real-world datasets while maintaining stability and scalability. We further extend our dispatching algorithm to include the concurrent option of rebalancing by providing artificial rebalancing requests to the same model. The extended algorithm achieves superior performance, up to 38.9\%, compared to a rebalancing heuristic. To foster future research and ensure reproducibility, our code is publicly available on GitHub: \url{https://github.com/tumBAIS/HybridMADRL-AMoD-Relocation}.

\section{Control Problem}\label{sec:control_problem}
We consider a stochastic multi-stage control problem introduced in \cite{Enders_2023_AMODDispatching_HybridMADRL}, in which a profit-maximizing central operator manages a (possibly autonomous) fleet of vehicles to serve stochastic customer requests that enter 
\begin{figure}[!ht]
  \centering
  \includegraphics[width=\textwidth]{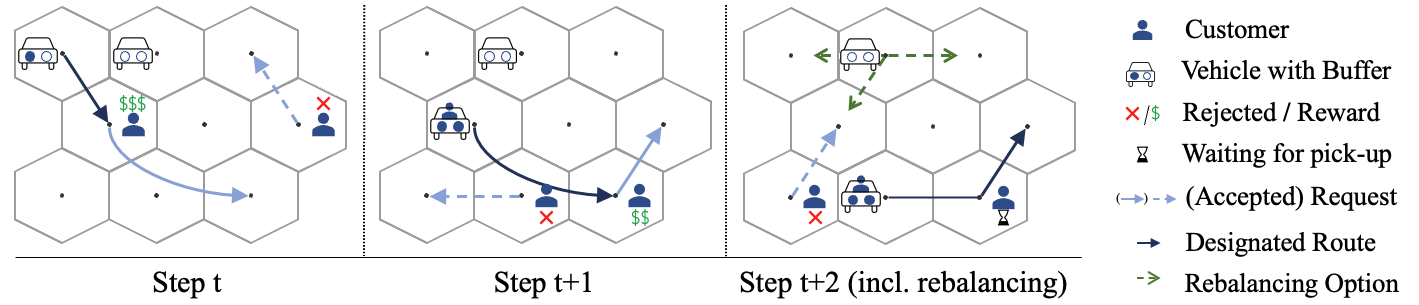}
  \caption{Schematic visualization of vehicle dispatching (and rebalancing) over time.}
  \label{fig:control_problem}
\end{figure}
the system over time, see \Cref{fig:control_problem}.
The operator accepts or rejects requests and dispatches vehicles accordingly. These decisions must be made in real time such that shifting rejected requests into the future is not possible. In an extension of this basic dispatching problem, we allow for a rebalancing decision, where the operator can redistribute idling vehicles within the operating area to proactively match vehicle supply and anticipated customer demand. We formalize the underlying control problem as a \gls{acr:mdp} in the following.

\textbf{Preliminaries:} 
We consider discrete time steps $t$ within a time horizon $\mathcal{T} = \{0, 1, \dots, T\}$, and represent the grid-based operating area as a graph $\mathcal{G} = (\mathcal{V}, \mathcal{E})$ with edge weights $e^d \in \mathbb{R}_{>0}$ denoting the distance and $e^t \in \mathbb{N}$ denoting the number of time steps required to traverse an edge $e \in \mathcal{E} \subseteq \{\{v,u\} \,|\, u,v \in \mathcal{V} \, \land \, v \neq u \}$. Vertices $v \in \mathcal{V}$ represent the centers of zones in the operating area. The neighbors of a vertex $v$ are given by $\mathcal{N}_{\mathcal{G}}(v) = \{u\;|\;(u,v) \in \mathcal{E}\}$.
If the operator accepts a request, the customer must be picked up within a given maximum waiting time $\omega^{\mathrm{max}} \in \mathbb{N}_0$.
In each time step, a variable number of requests $F_t$ appear in the system, and the operator needs to take simultaneous decisions over a batch of requests with a fixed-size fleet of autonomous vehicles $K$. 

\textbf{States:}
We denote the system state at time step $t \in \mathcal{T}$ by \linebreak $\mathbf{s}_t = (t, (\mathbf{k}_{j,t})_{j \in \{1, \dots, K\}}, (\mathbf{f}_{i,t})_{i \in \{1, \dots, F_t\}})$, with $K$ representing the number of vehicles $\mathbf{k}_{j,t}$, $j \in \{1, \dots, K\}$ and $F_t$ being the number of new fleet requests $\mathbf{f}_{i,t}$, $i \in \{1, \dots, F_t\}$.
We describe vehicles $\mathbf{k} = (p, \tau, \mathbf{f}^1, \mathbf{f}^2)$, with a position $p \in \mathcal{V}$, a number of time steps $\tau \in \mathbb{N}_0$ left to reach this position and a vehicle-specific request buffer $\mathbf{f}^1, \mathbf{f}^2$, which are tuples. 
Here, $p$ can either be the current vertex for an idling vehicle or the next vertex on the vehicle's route if it is traveling. To account for realistic trip lengths and maximum waiting times, we restrict the number of requests pre-assigned to a vehicle to two. An idling vehicle is characterized by an empty request buffer, i.e., $\mathbf{f}^1 = \emptyset$ as the request buffer fills up in a first-in-first-out fashion.
We label the position of vehicle $\mathbf{k}_{j,t}$ as \( p_{j,t} \) and label the other attributes of the vehicle in the same notation.
The set $\mathcal{F}_t = \mathcal{C}_t \cup \,\mathcal{B}$ contains all available requests in time step $t$ with $F_t = |\mathcal{F}_t|$ as the count of available requests, incl. variable customer requests $\mathcal{C}_t$ and time-persistent (re-)balancing requests $\mathcal{B}$.
If the operator does not allow vehicle rebalancing, no rebalancing requests are available, i.e., $\mathcal{B} = \emptyset$ and $\mathcal{F}_t = \mathcal{C}_t$.
We describe requests $\mathbf{f} = (\omega, o, d)$ with a waiting time $\omega \in \mathbb{N}_0\,\cup\,\emptyset$ tracking the elapsed time from request placement to pick-up, an origin $o \in \mathcal{V}$, and a destination $d \in \mathcal{V}$. To enable the representation of customer and rebalancing requests based on the same tuple $\mathbf{f}$, rebalancing requests obtain an empty waiting time ($\omega = \emptyset$) and an origin corresponding to their destination~($o = d$).
To allow for rebalancing, one creates a set of rebalancing requests $\mathcal{B} = \{(\emptyset, v, v) \;|\;v\;\forall\;\mathcal{V}\}$ to consider all vertices as rebalancing destinations.
Moreover, the operator cannot rebalance multiple vehicles to the same zone in one time step, as we create only one rebalancing request for each zone. However, one can easily relax this constraint by creating multiple rebalancing requests for each zone. 
To reduce the possible action space and enforce solely reasonable rebalancing movements, we define a vehicle-specific rebalancing request set $\mathcal{C}_t \cup \mathcal{B}_{j,t}$. We consider two variants to obtain this set for individual vehicles that permit equivalent rebalancing behavior:
(i) a vehicle $j$ can rebalance to any zone except its own: $\mathcal{B}_{j,t} = \{(\emptyset, v, v) \;|\;v\;\forall\;\mathcal{V} \setminus p_{j,t}\}$,
(ii) a vehicle $j$ can rebalance to its adjacent zones: $\mathcal{B}_{j,t} = \{(\emptyset, v, v) \;|\;v\;\forall\;\mathcal{N}_{\mathcal{G}}(p_{j,t})\}$.

\textbf{Actions:} 
The operator takes one decision $a_{i,t}$ per request $\mathbf{f}_{i,t}, i \in \left\{1, \dots, F_t\right\}$. The operator can either reject it $(a_{i,t} = 0)$ or assign it to a vehicle $(a_{i,t} = j)$, which is only possible if the vehicle has a free place in its request buffer and the request is a customer request, i.e., $\mathbf{f}^2_{j,t} = \emptyset \land \mathbf{f}_{i,t} \in \mathcal{C}_t$. 
A vehicle can also rebalance $(a_{i,t} = j)$ under the condition that it idles and the request is a vehicle-specific rebalancing request, i.e.,  $\mathbf{f}^1_{j,t} = \emptyset \land \mathbf{f}_{i,t} \in \mathcal{B}_{j,t}$. 
Each vehicle can be assigned to at most one request per time step. 
The action space of the central operator reads
\begin{equation}\label{eq:action_space}
\begin{split}
    \mathcal{A}(\mathbf{s}_t) = &\biggl\{ \Big(a_{1,t}, \dots, a_{F_t,t}\Big) \; \Big| \;
    (a_{i,t} = 0) \;\lor
    (a_{i,t} = j \land \mathbf{f}^2_{j,t} = \emptyset \land \mathbf{f}_{i,t} \in \mathcal{C}_t) \;\lor\\
    &(a_{i,t} = j \land \mathbf{f}^1_{j,t} = \emptyset \land \mathbf{f}_{i,t} \in \mathcal{B}_{j,t})
    \forall \; i \in \{1, \dots, F_t\} \land j \in \{1, \dots K\},\\
    &\sum_{i=1}^{F_t} \mathds{1}(a_{i,t} = j) \leq 1, \; \forall \; j \in \{1, \dots, K\} \biggl\}
\end{split}
\end{equation}

\textbf{Transition:}
First, we describe the action-dependent transition from the pre-decision state $\mathbf{s}_t$ to the post-decision state $\mathbf{s}_{t+}$. 
A reject decision does not alter the state. If a (rebalancing) request is assigned to a vehicle, the request is appended to that vehicle's request buffer.

The transition from the post-decision state to the next system state $\mathbf{s}_{t+1}$ is independent from the action and follows system dynamics.
If a vehicle picks up a customer at its origin vertex, we reset the waiting time of that request.
If a vehicle finds itself at a vertex and is serving a request, the position of the vehicle is updated based on the route taken. Vehicles take the shortest route to their next destination.
If a vehicle drops off a customer before the next decision is made, the request buffer removes the old request and shifts the content of the second request buffer position into the first.
The waiting time for accepted but not yet served requests increases by the time step increment. Rejected requests are dropped and leave the system immediately. New requests appear according to an unknown distribution.

\textbf{Rewards:} 
The operator aims to maximize profits. As autonomous vehicles have negligible fixed costs after an initial investment, the operational costs $c \in \mathbb{R}_{\leq 0}$ and trip fares are based on the driven distance. 
Each request is billed individually, depending on whether a car has picked up the customer on time. The revenue per request $\mathrm{rev}(\mathbf{f}) \in \mathbb{R}_{> 0}$ is based on the operator's pricing model and the total revenue per time step over all vehicles depends on the post-decision state
\begin{equation*}
    \text{Rev}(\mathbf{s}_{t+}) = \sum_{j=1}^{K} \mathds{1}(\mathbf{f}^1_{j,t+} \neq \emptyset \land \tau_{j,t+} = 0 \land p_{j,t+} = o(\mathbf{f}^1_{j,t+}) \land \omega(\mathbf{f}^1_{j,t+}) \leq \omega^{\mathrm{max}} ) \; \mathrm{rev}(\mathbf{f}^1_{j,t+}).
\end{equation*}

Regardless of whether customers are on board, the vehicle incurs variable costs denoted by $c \in \mathbb{R}_{\leq 0}$ per unit of distance traveled. Based on the vehicle's route from the post-decision position $p_{j,t+}$ to the next pre-decision position $p_{j,t+1}$, the total cost per time step results to

\begin{equation*}
    \text{Cost}(\mathbf{s}_{t+}) = c \,\sum_{j=1}^{K} \mathds{1}(\mathbf{f}^1_{j,t+} \neq \emptyset \land \tau_{j,t+} = 0)\; {(p_{j,t+}, p_{j,t+1})}^{d}.
\end{equation*}

The total profit at time $t+$ is $\mathrm{Profit}(\mathbf{s}_{t+}) =\mathrm{Rev}(\mathbf{s}_{t+})-\mathrm{Cost}(\mathbf{s}_{t+})$.
The post-decision state $\mathbf{s}_{t+}$ is a function of the pre-decision state $\mathbf{s}_t$ and the taken action $\mathbf{a}_t\in\mathcal A\left(\mathbf{s}_t\right)$. Thus, we write $\mathrm{Profit}(\mathbf{s}_{t+})=\mathrm{Profit}(\mathbf{s}_t, \mathbf{a}_t)$.

The centralized operator of the \gls{acr:amod} fleet aims to find a policy $\pi(\mathbf{a}_t | \mathbf{s}_t)$ maximizing the expected total profit over $\mathcal{T}$, based on the initialized state $\mathbf{s}_0$:

$$
    \text{Profit}(\mathbf{s}_0) = \max_{\pi} \mathbb{E}_{(\mathbf{s}_t, \mathbf{a}_t) \sim \pi} \left[ \sum_{t=0}^{T-1} \text{Profit}(\mathbf{s}_t, \mathbf{a}_t) \Big| \mathbf{s}_0 \right].
$$


\section{Methodology}\label{sec:methodology}

\Cref{fig:model} gives an overview of our algorithm's architecture. 
To derive an efficient and scalable algorithm, we employ multi-agent \gls{acr:sac} to keep the action space tractable. SAC’s entropy‑regularized policy yields superior stability and sample efficiency over other \gls{acr:rl} algorithms.
In this algorithm's architecture, vehicles represent agents, which use the same actor neural network and share parameters for efficiency. All agents evaluate all available requests and assign them a weight. We mask these weights to ensure feasibility constraints, and create a weighted bipartite graph with vehicle agents and requests representing vertices accordingly. 
We use this weighted matching to obtain an optimal vehicle-to-request allocation, i.e., a coordinated action for all agents based on the set of single agent outputs. During training, we use per-agent rewards to avoid a credit assignment problem \cite{Chang_2003}. The critic learns from these rewards and provides per-agent values to the actor to update its policy. 
To train actor and critic, we use the discrete version of \gls{acr:sac}, i.e., \gls{acr:sacd} \cite{SACD}.
Our algorithm scales linearly with the number of requests one agent has to evaluate. 
The employment of the \gls{acr:sacd} and the combination with bipartite matching
\begin{figure}

\vspace{-0.5cm}  \centering\includegraphics[width=1\textwidth]{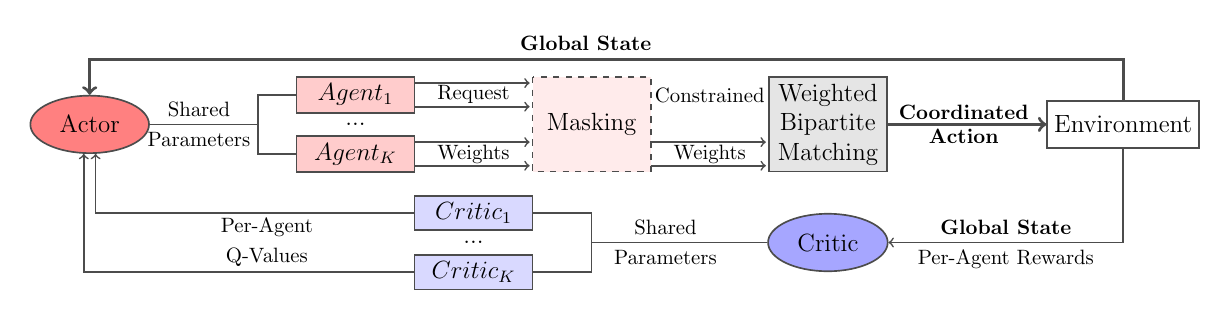}
  \caption{Overview of our algorithm's architecture, showing actor components (red), critic components (blue) and a combinatorial optimization component (gray).}
  \label{fig:model}
  \vspace{-0.25cm}
\end{figure}
are analogous to the algorithm proposed in \cite{Enders_2023_AMODDispatching_HybridMADRL}. 
Contrarily to~\cite{Enders_2023_AMODDispatching_HybridMADRL}, we use a vehicle-based parallel algorithm architecture, a masking procedure, an adapted critic loss function, and finally extend our algorithm to include rebalancing.


\subsection{Algorithm Architecture}\label{ssec:architecture}

We model each vehicle as an agent, see \Cref{fig:architecture}, as it enables the agents to build up reward trajectories corresponding to their actions, thereby fostering the learning process and accurate reward allocation. 
The agents' neural network architecture assesses each request based on the following input features: miscellaneous features $\textbf{m}_t$ contain general environment information, which is the same for all agents, e.g., the current time in the period; vehicle features $\textbf{v}_{j,t}$ are agent-specific, e.g., the current position and the request buffer; request features $\textbf{f}_{i,t}$ which we duplicate for all agents, e.g., the origin and the destination. 
We combine all features to one input-vector per vehicle and request. For the complete set of features, we refer to \Cref{sec:appendix_features}.

The actor assesses all currently available requests $\mathbf{f}_{i,t} \, , i \in \{1, \dots, F_t\}$ for one vehicle in parallel. 
To fix the input dimension of our agent’s neural network architecture, we pad the evaluated requests to have a constant size of $F_{max}$. To facilitate more simultaneous requests and speed-up computation for big instance sizes, one may limit $F_{max}$ via a broadcasting range, which gives only the closest requests to each vehicle.
We introduce an empty request $\mathbf{f}_{0,t}$ that allows a vehicle to reject unappealing requests. 
\begin{figure}[!bp]
\vspace{-0.5cm}

    \centering
    \includegraphics[width=1\textwidth]{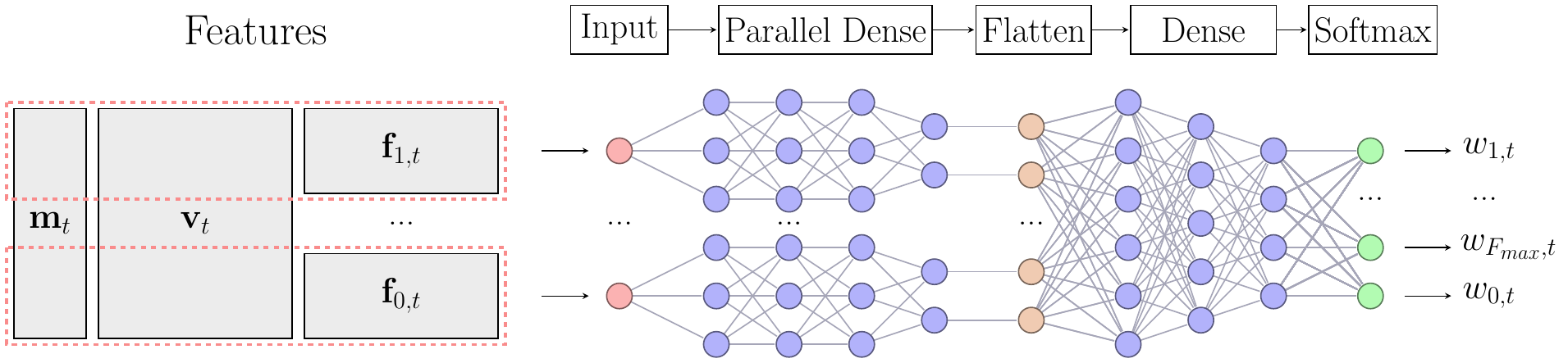}
    \caption{Vehicle-based agent architecture in which we combine input features to an input-vector (left) and use parallel neural networks (right).}

    \label{fig:architecture}
\end{figure}

We process each input vector separately through five parallel dense layers with identical parameters, forming a trainable multi-layer embedding. The two-dimensional output tensor is stacked and flattened into a one-dimensional vector, then evaluated in five dense layers. The final softmax activation outputs values $w
_{i,t} \in [0, 1]$ which we interpret as per-request weights. All agents share a common neural network architecture and parameters. The critic follows the same architecture and features as the actor but also receives the coordinated action as input. For algorithm specifications and final hyperparameters, we refer to \Cref{sec:appendix_methodology}. 

\subsection{Coordinated Critic Loss SAC-Discrete}\label{ssec:hybrid}
To derive the coordinated critic loss, we introduce the relevant notation. The per-request weights $w_{i,t}$ of each vehicle $j$ are only processed further if they exceed a threshold $\delta$, otherwise the vehicle rejects the request (cf. \Cref{app:post}). Solving this classic assignment problem with the masked actors' outputs as edge weights via bipartite matching, we get the coordinated action $\bar{a}$ of the fleet with per-vehicle action $\bar{a}_j$. We denote a transition by $(s,\bar{a},r,s')$, with environment state~$s$ and next state~$s'$ (denoted as $\mathbf{s}_{t+1}$ in the control problem). 

With these definitions, we are ready to derive the coordinated critic loss. Recall that \Gls{acr:sac} is an off-policy algorithm that incentivizes exploration by finding an optimal stochastic policy~$\pi^*$ via a maximum entropy objective. To this end, we parameterize the actor and critic network with parameters $\phi$ and $\theta$. We use two critic networks~$Q \in \{Q^1, Q^2\}$ to mitigate overestimation along with target networks with parameters~$\underaccent{\bar}{\theta}$ to ensure a stable target. We obtain the target parameters through an exponential moving average of the primary critic parameters. Then, the policy loss matching the control problem is
\begin{equation} \label{eg:actor}
    J_\pi(\phi) = \mathbb{E}_{(s,\bar{a},r,s')\sim \mathcal{D}} \biggr[ \sum_j \pi_\phi(a|s,j)^T \cdot \Bigl(\alpha \log (\pi_\phi(a|s,j))-\min_{m\in {1,2}}\Bigl\{Q_{\underaccent{\bar}{\theta}}^m(\bar{a}|s,j)\Bigl\}\Bigl)\biggr].
\end{equation}

We define~$r_j$ as the reward for agent~$j$ and sample transitions from a replay buffer $\mathcal{D}$. We define the policy's entropy as $\log(\pi_\phi(a|s,j))$, wherein $\alpha \in \mathbb{R}_{\geq0}$ controls the degree of exploration and $\gamma$ is the discount factor. Note that $\pi(a|s,j)$ represents the probability of vehicle $j$ deciding for a specific request. Hence, it corresponds to the weights $w$ in \Cref{fig:architecture} (see \Cref{ssec:architecture}), e.g., for a specific vehicle we get ${\pi(a=\mathbf{f}_{1,t}|s) = w_{1,t}}$.

The respective critic loss for our setting is
\begin{equation}
\begin{alignedat}{2}
    J_Q(\theta) &= \mathbb{E}_{(s,\bar{a},r,s')\sim \mathcal{D}} \biggr[ \sum_j \frac{1}{2} \biggl( Q_\theta (\bar{a}|s,j) - y_j \biggl)^2 \biggr]\text{, with } \\
    y_j &= r_j + \gamma \cdot \pi_\phi(a'|s',j)^T \cdot \Bigl(\min_{m\in {1,2}}\Bigl\{Q_{\underaccent{\bar}{\theta}}^m(a'|s',j)\Bigl\}-\alpha \log (\pi_\phi(a'|s',j))\Bigl)\label{eq:critic_loss}
\end{alignedat}
\end{equation}
as proposed in \cite{Enders_2023_AMODDispatching_HybridMADRL}.
As $\pi_\phi(a'|s',j)$ is the per-agent probability of taking per-agent action $a'$ in the next state, we refer to \Cref{eq:critic_loss} as a \textit{local loss}. The local loss does not accurately reflect the probability of agent $j$ executing $a'$ given $s'$. The local loss is only accurate, if $a'$ is executed (cf. \Cref{app:env_discussion}).
However, each agent executes an action $\bar{a}_j$ based on $\bar{a}$ which differs from $a'$ due to discrete assignments and feasibility constraints.  
To obtain a critic loss function that accurately reflects the coordinated action, we define a modified $\pi_\phi(a'|s',j)$, which we denote by $\bar{\pi}_\phi(a'|s',j)\in\{0,1\}$ that gets $\bar{\pi}_\phi(a'|s',j) = 1$ if $a' = \bar{a}_j$ and zero otherwise. Note that the distribution of $\bar{\pi}$ is degenerate as its support reduces to $\bar{a}_j$. We use the modified $\bar{\pi}_\phi(a'|s',j)$ to derive a new target $y_j$.

\begin{proposition}\label{prop:error}
Let us denote by $\bar{y}_j$ the per-agent target based on $\bar{\pi}_\phi(a'|s',j)$.
Then,
\begin{equation}
    \bar{y}_j = r_j + \gamma\, \min_{m\in {1,2}}\Bigl\{Q_{\underaccent{\bar}{\theta}}^m(\bar{a}_j|s',j)\Bigl\}. 
    \label{eq:adjusted_critic_target} 
\end{equation}
\end{proposition}

For a proof of \Cref{prop:error}, we refer to \Cref{proof:error}.
The adjusted critic loss function now reads 
\begin{align}
    J_Q(\theta) &= \mathbb{E}_{(s,\bar{a},r,s')\sim \mathcal{D}} \biggr[ \sum_j \frac{1}{2} \biggl( Q_\theta (\bar{a}|s,j) - \bar{y}_j \biggl)^2 \biggr]. \label{eq:critic_loss_adjusted}
\end{align}
and reflects a \textit{coordinated loss}, as we compute the critic loss using the agent-specific action which we derived from the coordinated action $\bar{a}$.
Herein, we obtain $\bar{y}_j$ according to \Cref{eq:adjusted_critic_target}. 
The actor loss does not require adaptation as it learns from the critic's predictions, which, in turn, already result from the correct updates based on \Cref{eq:critic_loss_adjusted}.

\section{Experimental Design}\label{sec:exp}

We follow \cite{Enders_2023_AMODDispatching_HybridMADRL,Heiko}, and use the well-established NYC Taxi dataset \cite{nycdata} to benchmark our model against state-of-the-art algorithms. We employ a hexagonal grid for spatial discretization to balance real-world similarity with computational efficiency and differentiate two settings: a smaller 11 zone (larger 38 zone) setting with \SI{500}{\meter} (\SI{1000}{\meter}) distance between adjacent zones in lower Manhattan. For the 11 small zones, we use the dataset without modifications. In the case of the 38 large zones, we apply a downscaling factor of 20 to reduce the number of requests, ensuring a comparable fleet size assessment in both scenarios.
We consider daily one-hour intervals as episodes, spanning from 8:30 to 9:30 a.m. with a time step size of one minute. Throughout the year 2015, excluding holidays and weekends, we include a total of 245 dates, using 200 for training, 25 for validation, and 20 for testing purposes. We configure the waiting time to be 5 (10) minutes. The cost parameter determines the cost per distance driven and is set to make round trips profitable with a 20\% margin. 

We consider two experiments: pure dispatching as well as integrated dispatching and rebalancing. For the experiment with pure dispatching, we model instances with supply shortage, as operators in dispatching scenarios with infinite supply can fulfill any request immediately, making a myopic policy sufficient. Thus, we set a maximum of 12 (20) requests per time step and deploy 12 (50) vehicles. 
For the experiment integrating dispatching and rebalancing, we model a supply overhang as rebalancing is only valuable if vehicles idle. 
We integrate rebalancing requests in our algorithm's architecture with artificial requests to imitate rebalancing as a dispatching decision, as explained in \Cref{sec:control_problem}. To leverage knowledge about the environment and enhance training, we alter the training rewards of rebalancing actions. The agent receives a fare (positive or negative) depending on the vehicle population of the current and the target zone.
We refer to \Cref{app:RebalancingRequests} for further details on the rebalancing requests and their training rewards.
For the 38 large zones, we generate rebalancing requests to
neighboring zones only (see \Cref{sec:appendix_rebalancing}). Thus, we set a maximum of 6 (10) requests per time step and deploy 24 (120) vehicles. We add a constraint that requires vehicles to be in the zone or en route to the origin zone of the request to foster rebalancing further. Moreover, we study the sensitivity w.r.t. the number of vehicles by considering additional instances with $\pm\,50\%$ vehicles. Finally, we create an additional dataset by shifting the originally rather homogeneous and well-balanced NYC Taxi dataset towards an imbalanced request distribution. The modified dataset has one (two) distinct clusters of departing customers, and we refer to it as ``Clustered'' hereafter. For details on the setup, evaluated scenarios, and hyperparameters we refer to~\Cref{sec:appendix_experiment}.

\section{Results and Discussion}
In the following, we analyze the performance of our algorithm focusing on a pure dispatching setting (\Cref{sec:dispatching}), as well as a setting that includes rebalancing decisions (\Cref{sec:rebalancing}). In the first part, we denote by $V_{\mathrm{local}}$ our algorithm using a local critic loss function (cf. Equation \ref{eq:critic_loss}) and by $V_{\mathrm{coord.}}$ our algorithm using a coordinated critic loss function (cf. Equation \ref{eq:critic_loss_adjusted}). We benchmark our dispatching performance against two algorithms: an algorithm using request-vehicle combinations as in \cite{Enders_2023_AMODDispatching_HybridMADRL} (\textit{RVC}) and a greedy policy (\textit{Greedy}). The \textit{Greedy} algorithm weighs every request according to its profitability before matching. In the second part, we show the capabilities of our integrated dispatching and rebalancing algorithm by extending $V_{\mathrm{coord.}}$ with rebalancing ($V_{\mathrm{ext.}}$). 
Here, we benchmark against \textit{Greedy}, our algorithm without rebalancing and a rebalancing heuristic (\textit{Heuristic}) that distributes the number of vehicles per zone equally. For a detailed description of all benchmarks, see \Cref{appendix_benchmarks}.

\subsection{Dispatching}\label{sec:dispatching}
\Cref{fig:abl_val} shows the validation results of all four algorithms for dispatching. \textit{RVC} and $V_{\mathrm{coord.}}$ exhibit a stable performance throughout training in both settings. \textit{RVC} converges to the same reward as \textit{Greedy}, while $V_{\mathrm{coord.}}$ converges to a significantly higher reward. The validation results of $V_{\mathrm{local}}$ show stability issues, especially for 38 zones. For the larger setting, $V_{\mathrm{coord.}}$ requires more training steps until convergence as the complexity rises with more vehicles and requests. The validation reward across training episodes of $V_{\mathrm{local}}$ demonstrates that our algorithm exhibits a high variance and a convergence below \textit{Greedy}, when neglecting the coordinated critic loss function. These observations show that the work of \cite{Enders_2023_AMODDispatching_HybridMADRL} improves stability and performance by employing one agent per request-vehicle combination, which improves performance even when using a local loss. While a vanilla implementation of our parallel vehicle-based algorithm is inferior to such an approach, utilizing our algorithm with a coordinated critic loss allows for significant improvements as the combination of a coordinated critic loss and a vehicle-based agent representation better reflects problem dynamics and consequently eases the learning.

\begin{figure}[!tp]
    \begin{subfigure}{0.45\textwidth}
        \includegraphics[height=1.3in]{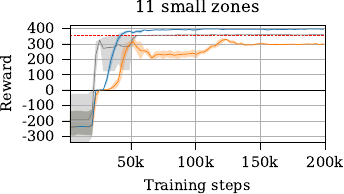}
    \end{subfigure}%
    \hspace{0.5cm}
    \begin{subfigure}{0.45\textwidth}
        \includegraphics[height=1.3in]{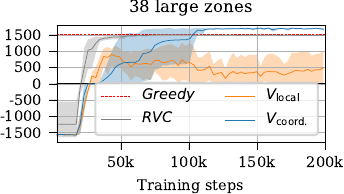}
    \end{subfigure}
    \caption{Average validation rewards of all four algorithms on both NYC settings. The shadowed area is the minimum and maximum value across 5 seeds.}
    \label{fig:abl_val}
\end{figure}

\Cref{table:abl} and \Cref{fig:abl:box} show each algorithm's improvement in test performance compared to \textit{Greedy} corresponding to the validation results in \Cref{fig:abl_val}. $V_{\mathrm{coord.}}$ outperforms \textit{Greedy} significantly, \textit{RVC} exhibits a test result close to \textit{Greedy} and $V_{\mathrm{local}}$ performs worse than \textit{Greedy}. The strong performance of $V_{\mathrm{coord.}}$ highlights the significance of employing a coordinated critic loss. In contrast, the low performance of $V_{\mathrm{local}}$ indicates that the error from using a local loss function substantially hinders the search for a stable and performant policy, which increases with a higher number of vehicles. \textit{RVC} evaluates one request-vehicle combination at a time, effectively splitting the vehicle into multiple agents. Thus, \textit{RVC} lacks the ability to perform anticipatory planning for individual vehicle routes as request-vehicle combinations only exist in the current time step. 
Consequently, \textit{RVC} inadequately accounts for the subsequent state $s’$ and the error resulting from bipartite matching, as it employs the current state for target retrieval, leading to a biased estimate that does not accurately predict the future state, yielding near \textit{Greedy} results.
Our vehicle-based architecture, on the other hand, is able to form a reward trajectory which is necessary for long-term planning. As the algorithm evaluates all requests and receives a reward corresponding to only its actions, it gains a correct estimate of the next state's value and can therefore optimize its subsequent dispatching decisions.

\begin{figure}
    \begin{subfigure}{0.39\textwidth}
        \includegraphics[height=1.05in, trim={0.25cm 0 0 0}, clip]{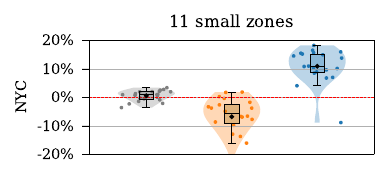}
    \end{subfigure}%
    \hspace{0.60cm}
    \begin{subfigure}{0.50\textwidth}
        \includegraphics[height=1.05in, trim={0.2cm 0 0 0}, clip]{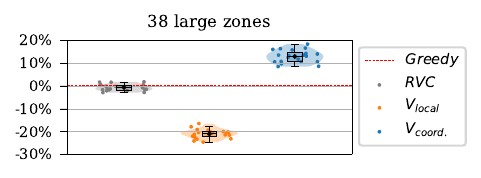}
    \end{subfigure}
    \caption{Test performances for dispatching (0\% indicates equal performance to \textit{Greedy}). Each dot represents the average of one test date across 5 seeds.}
    \label{fig:abl:box}
\end{figure}
\vspace{-20pt}

\begin{table}
\footnotesize
\parbox{.4\linewidth}{
\centering
\begin{tabular}{l|r|r}
\toprule
     Algo.  & 11 small & 38 large \\ \midrule
    \textit{Greedy}       &  350.9            &   1529.5          \\
    \textit{RVC}          &   $+1.9\%$  &  $-1.1\%$ \\
    $V_{\mathrm{local}}$  &  $-4.8\%$                      & $-21.0\%$    \\
    $V_{\mathrm{coord.}}$ &  +\textbf{12.9\%}   &  +\textbf{12.9\%} \\
\bottomrule
\end{tabular}
\vspace{0.25cm}
\caption{Performance improvement to \textit{Greedy} for dispatching across 5 seeds.}
\label{table:abl}
}
\hfill
\parbox{.545\linewidth}{
\centering
\begin{tabular}{l|r|r|r}
\toprule
    Metric             & \textit{Greedy}   & \textit{RVC}  & $V_{\mathrm{coord.}}$  \\ 
    \midrule
    Rejects to empty (\%)      & 18.3            & 18.2       &   44.7 \\
    Rejects to occupied (\%)   & 19.2           & 19.8       &   49.9  \\
    Pick-up distance (zones)    & 2.4               & 2.4           &  0.2  \\
    Waiting time (min)    & 3.8           & 3.7       &   2.5 \\
\bottomrule
\end{tabular}
\vspace{0.25cm}
\caption{Structural comparison of different algorithms on the 11 small zones setting (average across 5 seeds).}
\label{table:policy}
}
\end{table}

Note that, using request-vehicle combinations with global rewards instead of per-agent rewards also fixes the reward allocation problem. In case of global rewards, all agents have a shared goal, fostering coordination to optimize the collective outcome. For this reason, \cite{Heiko} extend \textit{RVC} to utilize global rewards. However, the test for \textit{RVC} with global rewards for our setting performed about the same as \textit{RVC} with per-agent rewards. Thus, indicating the challenge to learn from a single reward signal in combination with the lack of trajectories for the request-vehicle combinations.

To analyze the algorithms’ policies, we compare the metrics presented in \Cref{table:policy} for profitable requests. A request is profitable for the operator, if the fare is higher than the distance cost of any vehicle, and the vehicle can fulfill the request within the maximum waiting time. For each algorithm, we assess the share of rejected requests out of the number of profitable requests to gain insights into the algorithms' behaviors. The metric “rejects to empty” considers empty zones as destinations for the profitable requests and “rejects to occupied” counts zones with more than one vehicle located at that destination. We divide both metrics by the total number of requests to retrieve their ratio. The “pick-up distance” presents the distance driven to the pick-up zone from the current destination of the vehicle. Lastly, the “waiting time” is the average time a customer waits until pick-up. \textit{Greedy} and \textit{RVC} behave similarly, but $V_{\mathrm{coord.}}$ rejects with 44.7\% more than double the amount of profitable requests, keeps waiting times shorter, and reduces pick-up distance by a factor of ten. Hence, $V_{\mathrm{coord.}}$ learns a different policy in two ways. First, it better utilizes implicit rebalancing by preferring requests that allow heading towards empty zones over requests ending in occupied zones. Second, it rejects more low-profit requests to wait for those with higher profits, particularly favoring requests starting in the vehicle's current zone as indicated by the low pick-up distance. This leads to lower waiting times, as future customers rarely have to wait until the vehicle arrives at their location. Such a policy proves effective in instances with supply shortages, where the opportunity cost of rejecting requests is low.

\subsection{Integrated Dispatching and Rebalancing}\label{sec:rebalancing}

\Cref{fig:reb} shows the improvement in test performance compared to \textit{Greedy} for integrated dispatching and rebalancing. $V_{\mathrm{coord.}}$ performs close to \textit{Greedy} in all settings, whereas \textit{Heuristic} and $V_{\mathrm{ext.}}$ perform significantly better than \textit{Greedy}, especially on the Clustered dataset. The low performance improvement of $V_{\mathrm{coord.}}$ compared to dispatching-only bases on the oversupply of vehicles in the tested instances, for which \textit{Greedy} is inherently strong. Note that for the Clustered dataset, \textit{Heuristic} and $V_{\mathrm{ext.}}$ are able to outperform \textit{Greedy} by 96.2\% (51.3\%) and 101.0\% (111.5\%) respectively. Thus, our algorithm performs 39.8\% better than \textit{Heuristic} on the larger setting, but only 1.3\% on the smaller one. The result of \textit{Heuristic} for the smaller setting indicates that the rule-based rebalancing policy matches the distribution and frequency of requests for this instance size well.

\begin{figure}[!bp]
\vspace{-0.5cm}
  \centering
  \includegraphics[width=\textwidth]{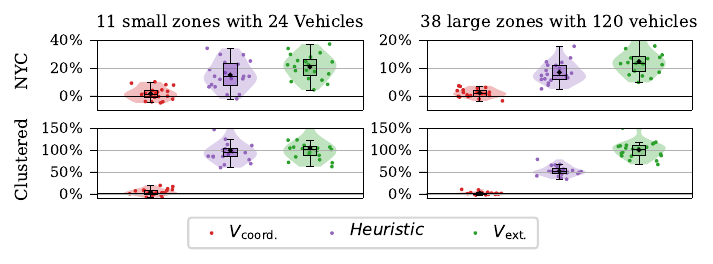}
  \caption{Test performances on both scenarios and datasets for integrated  dispatching and rebalancing compared to \textit{Greedy}. Each dot represents the average of one test date across 5 seeds.}
  \label{fig:reb}
\vspace{-0.5cm}
\end{figure}

\begin{table}
\footnotesize
\centering
\begin{tabular}{l|l|rrr|rrr}
\toprule
    & Zones & \multicolumn{3}{c|}{\textbf{11 small zones}} & \multicolumn{3}{c}{\textbf{38 large zones}} \\ \midrule
    Setting & Algo. \textbackslash \, $K$ & \#12 & \#24 & \#36 & \#60 & \#120 & \#180 \\ 
    \midrule
        \multirow{4}{*}{\textbf{NYC}} & \textit{Greedy} & 297.3 & 418.3 & 479.6 & 1594.4 & 2282.7 & 2674.1 \\
            & $V_{\mathrm{coord.}}$ & $+2.0\%$ & $+1.7\%$ & $+1.6\%$ & $+1.4\%$ & $+1.0\%$& $+0.7\%$ \\ 
            & \textit{Heuristic} & $+6.4\%$ & $+14.2\%$ & $+10.6\%$ & $+5.6\%$ & $+8.3\%$ & $+8.4\%$ \\ 
            & $V_{\mathrm{ext.}}$ & +\textbf{9.9\%} & +\textbf{20.0\%} & +\textbf{21.3\%} &+\textbf{6.4\%}& +\textbf{12.3\%} & +\textbf{10.1\%} \\  \midrule
         \multirow{4}{*}{\textbf{Clustered}} & \textit{Greedy} & 137.2 & 198.8 & 241.8 & 429.3& 688.5 & 882.6 \\
            & $V_{\mathrm{coord.}}$ & +1.2\% & +3.1\% & +0.8\% & +1.4\% & +1.0\% & +0.7\% \\ 
            & \textit{Heuristic} & +74.1\% & +96.2\% & +70.5\% &+66.5\% & +51.3\% & +48.6\% \\ 
            & $V_{\mathrm{ext.}}$ & +\textbf{85.6\%} & +\textbf{101.0\%} & +\textbf{96.7\%} & +\textbf{107.7\%}& +\textbf{111.5\%} & +\textbf{106.4\%}\\ 
\bottomrule
\end{tabular}
\vspace{0.25cm}
\caption{Performance improvement to \textit{Greedy} for dispatching and rebalancing to compare the impact of varying numbers of vehicles $K$ (ceteris paribus) across 5 seeds.}
\label{table:rebalancing}
\vspace{-0.75cm}
\end{table}

Furthermore, we explore the sensitivity of varying number of vehicles on operator returns, see \Cref{table:rebalancing}. Notably, our algorithm $V_{\mathrm{ext.}}$ demonstrates an overall increasing relative performance improvement in comparison to \textit{Heuristic} as the number of vehicles rises. This supports the scalability of our algorithm and indicates that $V_{\mathrm{ext.}}$ effectively capitalizes on past experiences, anticipating future requests and incorporating the spatial distribution, especially for the Clustered dataset. Determining the optimal fleet size is imperative for operators aiming to maximize their profits gained from their rebalancing policy. When the fleet size is small, the degree of freedom to improve is also small. When the fleet size is excessively large, rebalancing becomes less profitable in relative terms as \textit{Greedy} already achieves a comparable performance.

\section{Conclusion}

This work studies the fleet control problem of a profit-maximizing \gls{acr:amod} operator and offers a comprehensive solution, including the implementation code. We solve the dispatching problem by proposing a novel multi-agent \gls{acr:sacd} architecture, in which each agent first evaluates requests in parallel and combines them afterwards for a vehicle-based output. Thus, our algorithm ensures computational efficiency while optimizing its reward trajectory through long-term planning. We show an error in the critic's loss function and demonstrate how to accurately derive a coordinated loss for estimating future state values when combining multi-agent \gls{acr:sac} with a coordination layer, achieved via bipartite matching. By adjusting the critic loss function to a coordinated loss, we obtain a more accurate estimate of the next state's value, fostering the learning process of our agent. In addition, we extend our dispatching algorithm by incorporating concurrent rebalancing capabilities. Experimental results show that our approach surpasses state-of-the-art benchmarks while demonstrating stability across different instances. For dispatching, we outperform the closest benchmark by up to 12.9\% and for integrated dispatching and rebalancing by up to 38.9\%. In future work, we will extend the \gls{acr:amod} fleet control problem by covering charging and investigate how the adaptation to a coordinated critic loss function impacts the performance of \gls{acr:drl} models with combinatorial optimization in other multi-agent problem settings.

\begin{credits}
\end{credits}
%
%

%
\printbibliography
\pagebreak
\appendix



\section{Proof of \Cref{prop:error}}\label{proof:error}
\setcounter{proposition}{0}
\begin{proposition}
Let us denote by $\bar{y}_j$ the per-agent target using $\bar{\pi}_\phi(a'|s',j)$.
Then:
\begin{equation*}
    \bar{y}_j = r_j + \gamma \, \min_{m\in {1,2}}\Bigl\{Q_{\underaccent{\bar}{\theta}}^m(\bar{a}_j|s',j)\Bigl\} 
\end{equation*}
\end{proposition}
\begin{proof}
First, we express the target $\bar{y}_j$ based on $\bar{\pi}_\phi(a'|s',j)$ using the \gls{acr:sacd} target definition in Equation \eqref{eq:critic_loss} 
\begin{equation}
\bar{y}_j = r_j + \gamma \, \bar{\pi}_\phi(a'|s',j)^T \cdot \Bigl(\min_{m\in {1,2}}\Bigl\{Q_{\underaccent{\bar}{\theta}}^m(a'|s',j)\Bigl\}-\alpha \log (\bar{\pi}_\phi(a'|s',j))\Bigl). \label{eq:pi_adjusted_not_yet_developed}
\end{equation}
Now, we develop \eqref{eq:pi_adjusted_not_yet_developed} and start with the scalar product as the sum over all actions belonging to $\mathcal{A}$
\begin{equation}
\begin{alignedat}{2}
\bar{y}_j & =  r_j &&+ \gamma \, \sum\limits_{a' \in \mathcal{A}}  \biggl[ \bar{\pi}_\phi(a'|s',j) \Bigl(\min_{m\in {1,2}}\Bigl\{Q_{\underaccent{\bar}{\theta}}^m(a'|s',j)\Bigl\}-\alpha \log (\bar{\pi}_\phi(a'|s',j))\Bigl)\biggl] \\
&\stackrel{\text{(I)}}{=} r_j &&+ \gamma \bar{\pi}_\phi(\bar{a}_j|s',j) \Bigl(\min_{m\in {1,2}}\Bigl\{Q_{\underaccent{\bar}{\theta}}^m(\bar{a}_j|s',j)\Bigl\}-\alpha \log (\bar{\pi}_\phi(\bar{a}_j|s',j))\Bigl)\\ 
& &&+ \gamma \sum\limits_{a' \in \mathcal{A} \setminus \bar{a}_j}  \biggl[ \bar{\pi}_\phi(a'|s',j) \Bigl(\min_{m\in {1,2}}\Bigl\{Q_{\underaccent{\bar}{\theta}}^m(a'|s',j)\Bigl\}-\alpha \log (\bar{\pi}_\phi(a'|s',j))\Bigl)\biggl]\\
&\stackrel{\text{(II)}}{=} r_j &&+ \gamma \underbrace{\bar{\pi}_\phi(\bar{a}_j|s',j)}_{=\;1} \Bigl(\min_{m\in {1,2}}\Bigl\{Q_{\underaccent{\bar}{\theta}}^m(\bar{a}_j|s',j)\Bigl\}-\alpha \underbrace{\log (\bar{\pi}_\phi(\bar{a}_j|s',j))}_{=\;0}\Bigl)\\
& &&+ \gamma \sum\limits_{a' \in \mathcal{A} \setminus \bar{a}_j}  \biggl[ \underbrace{\bar{\pi}_\phi(a'|s',j) \min_{m\in {1,2}}\Bigl\{Q_{\underaccent{\bar}{\theta}}^m(a'|s',j)\Bigl\}}_{=\;0}-\bar{\pi}_\phi(a'|s',j)\alpha \log (\bar{\pi}_\phi(a'|s',j))\biggl]\\
&\stackrel{\text{(III)}}{=} r_j &&+ \gamma\, \min_{m\in {1,2}}\Bigl\{Q_{\underaccent{\bar}{\theta}}^m(\bar{a}_j|s',j)\Bigl\} - \gamma \sum\limits_{a' \in \mathcal{A} \setminus \bar{a}_j}  \biggl[\underbrace{\bar{\pi}_\phi(a'|s',j) \alpha \log (\bar{\pi}_\phi(a'|s',j))}_{=\;0}\biggl]\\
&= r_j &&+ \gamma\, \min_{m\in {1,2}}\Bigl\{Q_{\underaccent{\bar}{\theta}}^m(\bar{a}_j|s',j)\Bigl\}.
\end{alignedat}
\end{equation}
In the above derivation, we used the following reasoning:

(I) We divide the sum into terms depending on the action $\bar{a}_j$ and terms depending on the set of actions $\mathcal{A} \setminus \bar{a}_j$. Here, $\bar{a}_j$ is the agent-specific action derived from the coordinated action $\bar{a}$. 

(II) We recall that $\bar{\pi}_\phi(a'|s',j) = \begin{cases}
    1, \text{ if } a' = \bar{a}_j\\
    0, \text{ if } a' \neq \bar{a}_j\\
\end{cases}$, as the distribution $\bar{\pi}_\phi(a'|s',j)$ is degenerate. 

(III) We use the rule of L'Hôpital to compute $\bar{\pi}_\phi(a'|s',j)^T\alpha \log (\bar{\pi}_\phi(a'|s',j))$ for $\bar{\pi}_\phi(a'|s',j)=0$, if $a'\neq \bar{a}_j$. Let us denote $x=\bar{\pi}_\phi(a'|s',j)$, then the following holds
\begin{equation}
    \lim_{x\rightarrow 0}\;x\; \log(x) = \lim_{x\rightarrow 0}\; \frac{\log x}{\frac{1}{x}} \stackrel{\mathrm{\text{L'Hôpital}}}{=} \lim_{x\rightarrow 0}\; \frac{\frac{1}{x}}{-\frac{1}{x^2}} = \lim_{x\rightarrow 0}\; \frac{-x^2}{x} = \lim_{x\rightarrow 0}\; -x = 0.
\end{equation}
Hence, we set $\bar{\pi}_\phi(a'|s',j)^T \alpha \log (\bar{\pi}_\phi(a'|s',j)) = 0$.
\end{proof}

\subsection{Dependency on the Environment}\label{app:env_discussion}
We note here, that the usage of the unadjusted loss $y_j$ can lead to good results depending on the environment. Using the unadjusted loss in unconstrained environments, where the agent can execute its action directly, e.g., cooperative environments that determine only the reward depending on the action of others, can yield good results \cite{Iqbal_2019}. However, in our problem setting the matching algorithm introduces a competitive component leading to a constrained environment. In competitive environments, where the executed actions depend on the actions of others, using the per-agent action leads to an increasingly inaccurate $y_j$ with a higher number of agents. The increasing mismatch between per-agent actions and the actually executed actions, derived from the coordinated action, is shown by \cite{Lowe_2017} where they propose their multi-agent deep deterministic policy gradient algorithm. 

\section{Methodology} \label{sec:appendix_methodology}
In the following, we describe methodological details. First, we present our chosen features and architecture details of the employed neural networks. Second, we outline the hyperparameters used in our experiments. Third, we explain the masking, which we use to retrieve the weights for our bipartite matching. Fourth, we describe details regarding the definition of rebalancing requests and the corresponding reward mechanisms applied during the training process.

\subsection{Feature \& Neural Network Architectures} \label{sec:appendix_features}
In this section, we describe the features used for our new parallel architecture in detail. The input state for the first parallel dense layer is specific to the respective agent $j$ and considers all available requests $i \in \left\{1, \dots, F_t\right\}$. As we fix the number of requests per agent, we extend $F_t$ to $F_{max}$ through padding to have a static input size.
We use a spatial discretization to divide the operating area into a hexagonal grid. Each grid zone represents a central point with horizontal and vertical indices. We map all possible locations of the operating area to zones. The encoding of a location is based on the normalized indices of the respective zone.

We now list the specific environment information that we use as inputs to generate features for the neural network. The inputs differ slightly depending on whether rebalancing is enabled ($V_{\mathrm{ext.}}$) or whether only the critic evaluates the feature (Critic):
\begin{itemize}
    \item \textbf{Miscellaneous features} $\textbf{m}_t$ contain general environment information and are the same for all agents: 
    \begin{itemize}
        \item Time step in the episode, normalized to [0,1]
        \item Aggregated time steps needed for all vehicles to reach their final destination positions, normalized to [0, 1], which indicates the current level of activity within the fleet
        \item Count of requests placed since start of current episode, divided by the count of requests placed on average until the current time step, which indicates the observed demand compared to an average episode
    \end{itemize}
    \item \textbf{Vehicle features} $\textbf{v}_{j,t}$ are agent-specific:
    \begin{itemize}
        \item Normalized encoding of vehicle's end position, which is the current vertex if no request is assigned, otherwise destination of assigned request that will be served last
        \item Time steps to reach this position, divided by the maximum time between any two vertices in the graph representing the operating area
        \item Number of assigned requests, normalized to [0,1]
    \end{itemize}
    \item \textbf{Request features} $\textbf{f}_{i,t}$ are duplicated for all agents: 
    \begin{itemize}
        \item Normalized encoding of the request origin zone
        \item Normalized encoding of the request destination zone
        \item Distance from origin to destination on graph G, normalized by maximum distance to [0, 1]
        \item ($V_{\mathrm{ext.}}$) Rebalancing flag in $\{0, 1\}$
        \item ($V_{\mathrm{ext.}}$) Number of vehicles with final destination in pick-up zone, normalized to [0,1]
        \item ($V_{\mathrm{ext.}}$) Number of vehicles with final destination in drop-off zone, normalized to [0,1]
    \end{itemize}
    \item \textbf{Request-Vehicle features} $(\textbf{f}_{i,t},\textbf{v}_{j,t})$ are specified per request $i$ and vehicle $j$:
    \begin{itemize}
        \item Distance from vehicle position to request origin, normalized by maximum distance to [0, 1]
        \item (Critic) Waiting time flag in $\{0, 1\}$, which indicates whether a request can be picked up in time
    \end{itemize}
\end{itemize} 

Next, we describe the actor network. The general structure of the input processing resembles prior work by \cite{Enders_2023_AMODDispatching_HybridMADRL} using the attention mechanism and embeddings presented in \cite{Holler_2019_MODRebalancing_MADRL} and \cite{kullmann_attention}.
The request and vehicle input features serve as inputs to create the respective embeddings. 
Each embedding consists of a feedforward dense layer with 32 units and a \gls{acr:relu} activation.
The attention mechanism computes a context from each embedding. The global context is the concatenation of the request context and the vehicle context. As the number of requests are time-dependent, we use embeddings and contexts to represent variable-size inputs in a fixed-size global representation.
For each path of the parallel network (see \Cref{fig:architecture}), we combine the request and vehicle embeddings, the global context, and the specific request-vehicle input features. The parallel networks act as trainable multi-layer embeddings. We experience better performance when we shuffle the parallel inputs for all customer requests in order to train all subsequent nodes after the flattening equally.
We evaluate these inputs for five parallel dense layers with unit sizes of 512, 256, 128, 64, 32 and after flattening for six layers with unit sizes of 1024, 512, 256, 128, 64, 32. We apply L2 regularization with a coefficient of $10^{-4}$ to all layers. Generally, we use $\text{ReLU}$ activation for the feedforward layers and evaluate the final output on a softmax activation with $|F_{max}+1|$ units. 

The critic network architecture is similar to the actor network's architecture. However, we add the information regarding which requests were accepted after the coordination to the input features for the embeddings and the context calculation. To this end, we add a binary flag ${0,1}$ that denotes rejection and acceptance for requests. For the vehicle input features, we integrate the information of the origin and the destination of newly assigned requests. The critic output does not have an activation function.

\subsection{Hyperparameters}\label{sec:appendix_hyperparameters}
We mostly use the same hyperparameters as \cite{Enders_2023_AMODDispatching_HybridMADRL} for our algorithm. We train for 200,000 steps (300,000 for rebalancing on 38 large zones), update the network parameters every 20 steps, and test the performance of the current policy on the validation data every 2,880 steps (48 episodes). During the first 20,000 steps, we collect experience with a random policy and do not update the network parameters. For the next 30,000 steps, we add linearly declining noise to the actor's weights. For the critic loss, we use the Huber loss with a delta value of 10 instead of the squared error. We use gradient clipping with a clipping ratio of 10 for actor and critic gradients. We use the Adam optimizer with a learning rate of $3\times10^{-4}$. We sample batches of size 128 from a replay buffer with maximum size of 100,000 for the first experiment and 50,000 for the second experiment. We set the discount factor to 0.925. For the update of the target
critic parameters we use an exponential moving average with smoothing factor $5\times10^{-4}$. We tune the entropy coefficient individually per instance in the range of 0.2 to 0.6.

We conduct multiple training runs using five different random seeds and select the model with the highest validation performance across these runs. The selected model undergoes testing on the dataset, and the results presented in this paper reflect its performance on the test data.

\subsection{Edge weight masking}\label{app:post}
In Figure \ref{fig:architecture}, we use deterministic post-processing to mask the edge weights used in matching to adhere to our constraints and to determine the per-agent reject action. The masking step enables the agent to reject single requests by comparing the weights against a constant threshold. The selected threshold is the highest value at which the agent is able to accept all requests and also rank them against one another. Thus, a per-request weight must be greater than $\delta = 1 / (F_{max} + 1)$. By ranking we refer to the agent's ability to signal its wish to take multiple requests, but valuing them differently. During training, the agent differentiates between an active and a passive reject action. The active reject action occurs if the agent does not want any request. In this case, we will consider the $Q$-value and probability of the reject action when calculating the loss. For the passive reject action, the system does not use the $Q$-value and probability when calculating the loss, as the matching decided that the agent does not receive a request instead of the agent choosing to reject all its requests. Thus, the passive reject action is not the same as the active reject action and therefore the $Q$-value for actively rejecting cannot be used.

In \Cref{alg:post} we describe how we mask the weights obtained from the actor (cf. \Cref{fig:model}). We start with the plain weights from the actor as input data and change them to match the environment constraints and to retrieve the reject action. We first check whether the second place of the request buffer $\mathbf{f}^2_{j,t}$ is occupied (\Cref{line:post1}). If this is the case, we set all weights of the vehicle $w_{i,t}$ to 0 and a passive reject action $a_{0,t} = 0$ (\Cref{alg:post2}). If not, we check for each weight of the vehicle $w_{i,t}$ if is it higher than the previously mentioned threshold which we denote as $\delta$ (\Cref{alg:post5}). If this is the case, we keep the weight as-is (\Cref{alg:post6}) if this is not the case, we set it to 0 (\Cref{alg:post8}). Finally, we check for an active reject action by adding up all weights of one vehicle and check whether the sum corresponds to zero (\Cref{alg:post9}). If this is the case, we determine an active reject action $a_{0,t} = 1$ (\Cref{alg:post10}) and if not a passive reject action $a_{0,t} = 0$ (\Cref{alg:post12}). 

\begin{algorithm}[!htbp]
\SetKwInput{KwData}{Input Data}
\SetKwInput{KwResult}{Output Data}
 \raggedright
 \KwData{weights $w_{i,t}$ ; threshold $\delta = \frac{1}{F_{max}+1}$}
 \KwResult{weights $w_{i,t}$ ; reject action $a_{0,t}$}
 \uIf{$\mathbf{f}^2_{j,t} \neq \emptyset$\label{line:post1}}{
    $w_{i,t}$ = 0 ; $a_{0,t}$ = 0 \;
  \label{alg:post2}}
 \Else{
    \For{$i$ in $F_t$}{
    \uIf{$w_{i,t} > \delta$\label{alg:post5}}{$w_{i,t} = w_{i,t}$\label{alg:post6}}
    \Else{$w_{i,t}$ = 0\label{alg:post8}}}
    \uIf{ $ \sum_i w_{i,t} = 0$\label{alg:post9}}{$a_{0,t}$ = 1}\label{alg:post10}
    \Else{$a_{0,t}$ = 0}\label{alg:post12}}
 \caption{Masking per vehicle agent}\label{alg:post}
\end{algorithm}

\subsection{Rebalancing Requests}\label{app:RebalancingRequests}

To seamlessly integrate rebalancing with dispatching, we model rebalancing actions as requests. We constrain rebalancing requests and differentiate them from customer requests as follows:
\begin{enumerate}[label=(\alph*)]
    \item Vehicle $j$ can only accept rebalancing requests when it idles, i.e., when its request buffer is empty ($\mathbf{f}^1_{j,t} = \emptyset \land \mathbf{f}^2_{j,t} = \emptyset$). Therefore, only one rebalancing action can be undertaken at a time.
    \item The origin of rebalancing requests corresponds to their destination. Thus, a vehicle's approach to that destination is considered as rebalancing.
    \item Rebalancing request $i$ is only possible to a zone not corresponding to the current zone of vehicle $j$, i.e, $p_{j,t} \neq d_{i,t}$.
    \item Rebalancing requests must be accepted immediately and cannot be deferred, i.e., $\omega_{i,t} = \emptyset$.
    \item Rebalancing requests generate operational costs corresponding to the distance traveled by the vehicle that fulfills it.
    \item We incorporate a flag into the rebalancing request features to distinguish them from customer requests.
\end{enumerate}

As rebalancing requests incur only costs, every rebalancing action returns a negative reward when executed. The algorithm learns to identify the positive impact of rebalancing in future time steps via the Bellman recursion. However, we can speed up learning by incorporating additional reward signals during training. We do so by calculating an artificial reward signal for rebalancing, only used during training, that is added to the operational costs of rebalancing in order to incentivize good rebalancing behavior and to penalize bad behavior.

\Cref{alg:rebalancing_rewards} describes the calculation of the training reward signal for rebalancing requests. The variable $x_t^d(\mathbf{f}_i)$ in \Cref{alg2:destination} measures the number of vehicles in the destination zone $d$ of the rebalancing request $i$ at time $t$ normalized to a value between 0 (full zone) and 1 (empty zone). The variable $x_t^o(\mathbf{f}_i)$ calculates a similar proxy measure for the origin zone $o$ of the request. We distinguish between empty and full origin zones (\Cref{alg2:if}), as empty zones result in a negative reward signal (\Cref{alg2:origin1}) and full zones in a positive reward signal (\Cref{alg2:origin2}). We give the positive reward signal in \Cref{alg2:origin2} proportional to the emptiness of the destination as we multiply by $x_t^d(\mathbf{f}_i)$, and we increase the negative reward signal faster (indicated by factoring it with -2). Finally, we sum up $x_t^d(\mathbf{f}_i)$ and $x_t^o(\mathbf{f}_i)$, wherein we attribute a lower weight to $x_t^o(\mathbf{f}_i)$ (\Cref{alg2:final_reward}) to mitigate unnecessary rebalancing. Thus, we have a reward signal that accounts for the occupancy of both the destination zone and the origin zone, to give the complex state representation additional reward signals, which rebalancing would otherwise not receive. 

\begin{algorithm}[H]
    \setcounter{AlgoLine}{0} 
    \raggedright
    \SetKwInput{KwData}{Input Data}
\SetKwInput{KwResult}{Output Data}
    \KwData{number of vehicles at origin of the request $o_t(\mathbf{f}_i)$ ; number of vehicles at destination of the request $d_t(\mathbf{f}_i)$, cost parameter $c$, minimum distance between two vertices $l$, average number of vehicles per zone rounded up $\lceil v_t \rceil$, average number of vehicles per zone rounded down $\lfloor v_t \rfloor$}
    \KwResult{training signal $s_{j,t}$ for vehicle $j$}
    $x_t^d(\mathbf{f}_i) = 1 - \min\{1, \frac{d_t(\mathbf{f}_i)}{\lfloor v_t \rfloor}\}$\label{alg2:destination}\\
    \uIf{$o_t(\mathbf{f}_i)<v_u$\label{alg2:if}}{$x_t^o(\mathbf{f}_i) = \frac{(\lceil v_t \rceil - o_t(\mathbf{f}_i)+1)\times(-2)}{\lceil v_t \rceil}$\label{alg2:origin1}}
    \Else{$x_t^o(\mathbf{f}_i) = \frac{o_t(\mathbf{f}_i) - \lceil v_t \rceil}{\lceil v_t \rceil}\times x_t^d(\mathbf{f}_i)$\label{alg2:origin2}}
    $x_t(\mathbf{f}_i)= \frac{x_t^o(\mathbf{f}_i)}{2}+2\times x_t^d(\mathbf{f}_i)$\\
    $f_{j,t} = c \times l \times x_t(\mathbf{f}_i)$ \label{alg2:final_reward}
    \caption{Calculation of training signal for rebalancing during training}\label{alg:rebalancing_rewards}
\end{algorithm}

\section{Experiments}\label{sec:appendix_experiment}
In the following, we present additional details about our experiments. First, we provide information regarding the underlying datasets and the system configuration. Second, we elaborate on the three benchmark algorithms. Third, we offer insights into the rebalancing behavior exhibited by our policy.

\subsection{Datasets} \label{sec:appendix_datasets}

For our experiment, we use yellow taxi trip records in NYC from the year 2015 \cite{nycdata}, excluding weekends and holidays. We assume that request placement times coincide with the reported pick-up times in the dataset. We extract pick-up and drop-off longitude/latitude coordinates, restricting our experiment to trips where both coordinates fall on the main island of Manhattan. We use spatial discretization by employing hexagonal grids. We map each request to a pick-up and a drop-off zone based on the shortest distance of the requests pick-up and drop-off coordinates to the center of any zone. We exclude trips starting and ending in the same zone. The distances between neighboring zones are set at 459 meters for small zones and 917 meters for large zones. Travel time between two neighboring zones assumes two and five time steps based on realistic driving speeds. We construct a graph, wherein each zone represents a vertex and edges connect neighboring zones. Vehicles traveling between non-neighboring zones follow the shortest route.

We focus on two different subsets of zones within the NYC dataset, which define our simulated operating areas, i.e., the part of Manhattan in which our fleet operates. Hence, we only consider requests that have pick-up and drop-off locations within the operating area. For the 38 large zones, we downscale the trip data by a factor of 20, using every 20th request for simulation. This adjustment accommodates hardware limitations and results in an average of 360 requests per episode for the 11 small zones and 828 requests for the 38 large zones. The mean trip distance is larger for the 38 large zones, influencing the number of vehicles required.

\Cref{fig:data_11n} illustrates the smaller operating area of the NYC dataset, where colors represent the pick-up and drop-off frequency. Dark green represents zones with a majority of pick-ups, while dark red represents zones with a majority of drop-offs. Light yellow / green reflect zones where the number of pick-ups and drop-offs is almost equal. Although some of the 11 small zones exhibit a minor bias toward either drop-off or pick-up, it is not strongly pronounced. Therefore, we obtain a second dataset which maintains the temporal patterns of the NYC taxi dataset while altering the drop-off and pick-up locations to achieve more imbalanced request distributions. We do so by sampling pick-ups from a normal distribution around a cluster and sampling drop-offs towards its edges from another normal distribution. This modification allows for a more imbalanced spatial distribution of requests, depicted by dark green and dark red zones in \Cref{fig:data_11s}. This simulates realistic scenarios like, e.g., the end of a major sports event in the middle of the operating area. \Cref{fig:data_38} displays the operating area of the 38 large zones with its spatial distribution.

We assume a maximum waiting time of five minutes for the 11 small zones setting and ten minutes for the 38 large zones setting, with the increased waiting time due to the longer trip lengths and greater size of the operating area. To achieve a 20\% operating profit margin with empty driving to the pick-up location, we set the revenue at 5.00 USD per km, and operational costs at 2.00 USD per km.
\begin{figure*}[!tp]
    \centering
    \begin{subfigure}{.28\textwidth}
        \centering
        \includegraphics[width=\linewidth]{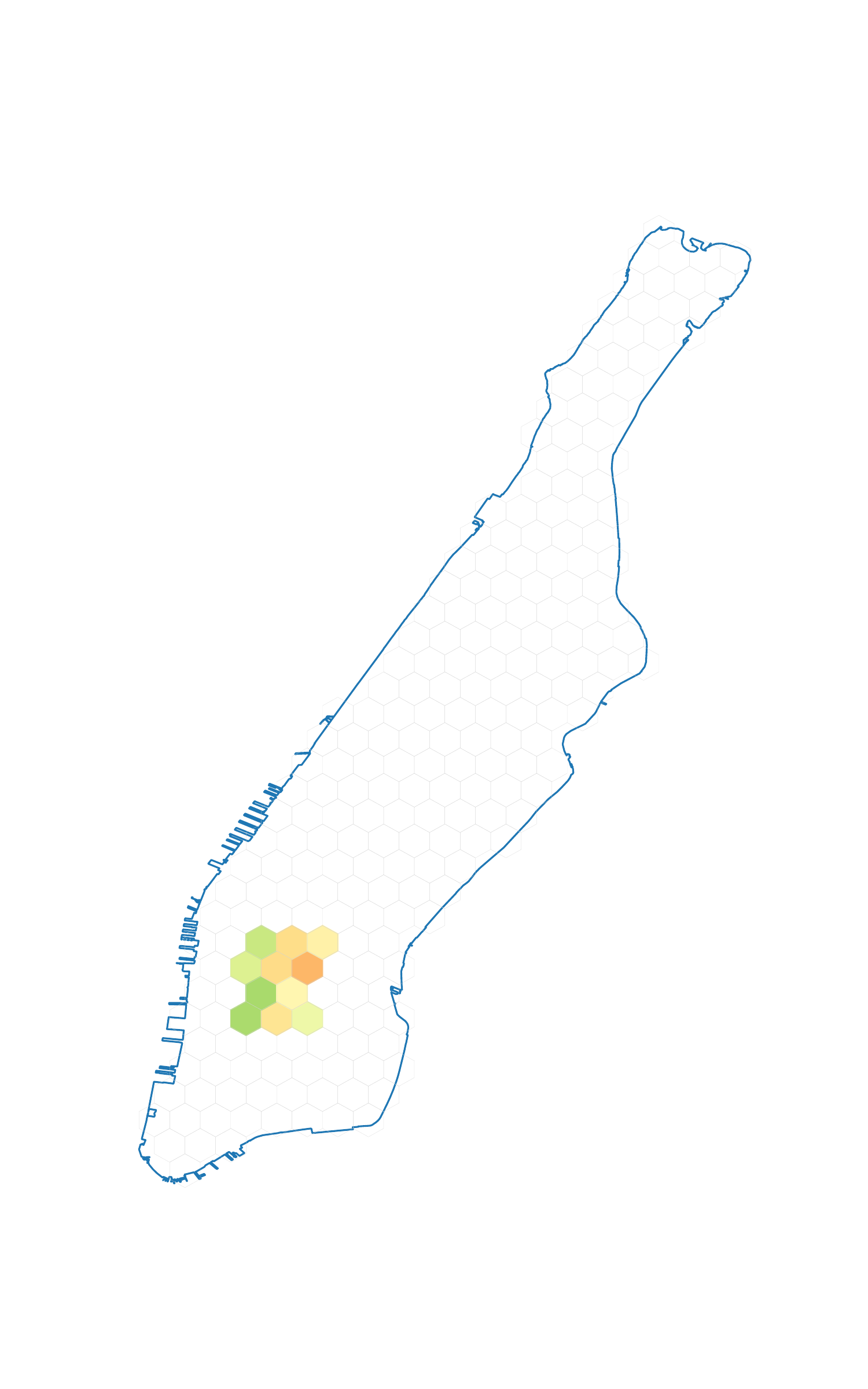}{}
        \caption{NYC - 11 zones}
        \label{fig:data_11n}
    \end{subfigure}
    \begin{subfigure}{.28\textwidth}
        \centering
        \includegraphics[width=\linewidth]{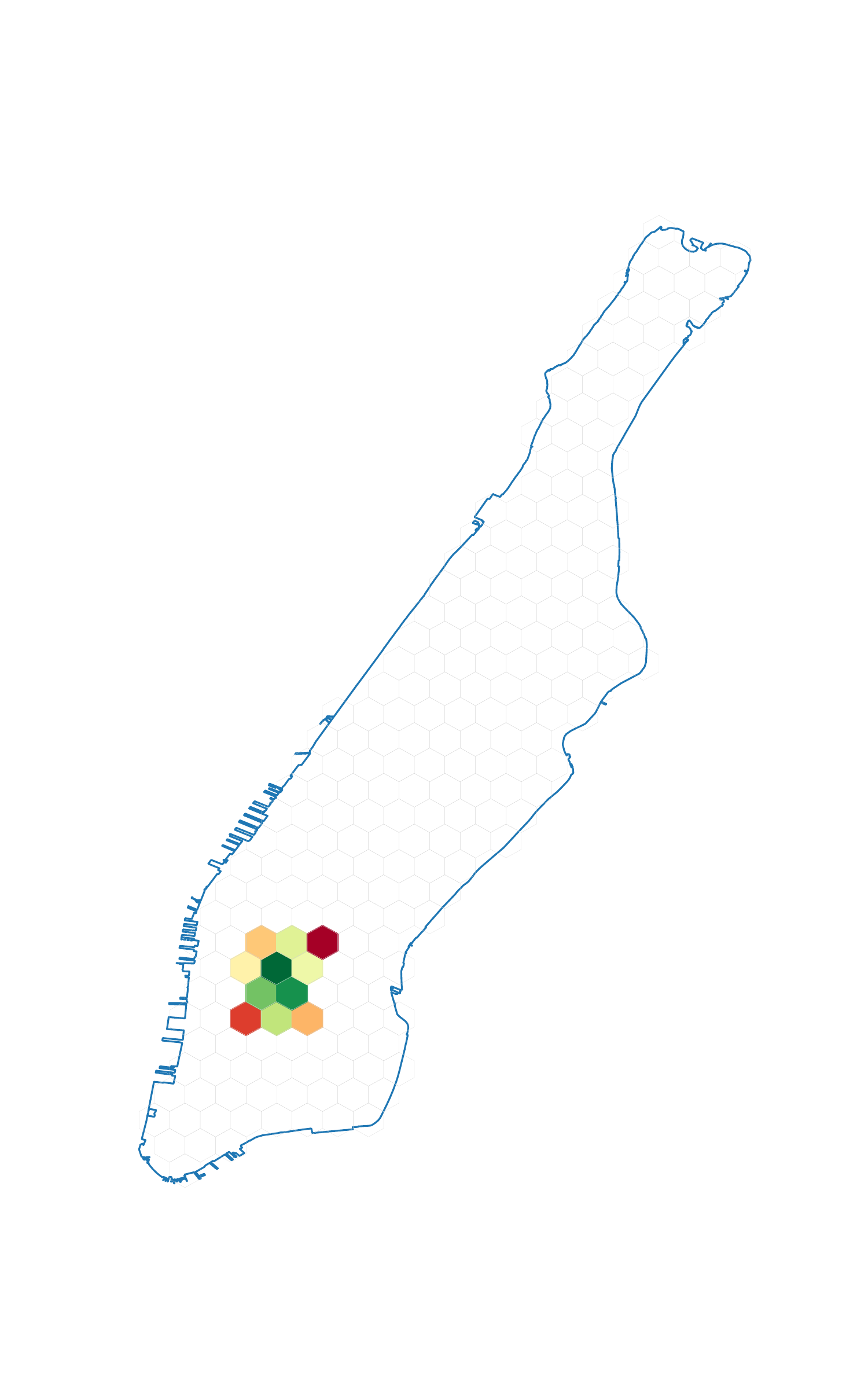}
        \caption{Clustered - 11 zones}
        \label{fig:data_11s}
    \end{subfigure}
\caption{The 11 zone operating area. Zones dominated by pick-ups are marked in green, zones dominated by drop-offs are marked in red, and zones with a balance between pick-ups and drop-offs are marked in light green/yellow.}
\label{fig:data}
\end{figure*}
\begin{figure*}[!htbp]
    \centering
    \begin{subfigure}{.28\textwidth}
        \centering
        \includegraphics[width=\linewidth,trim={0 2.5cm 0 2.5cm},clip]{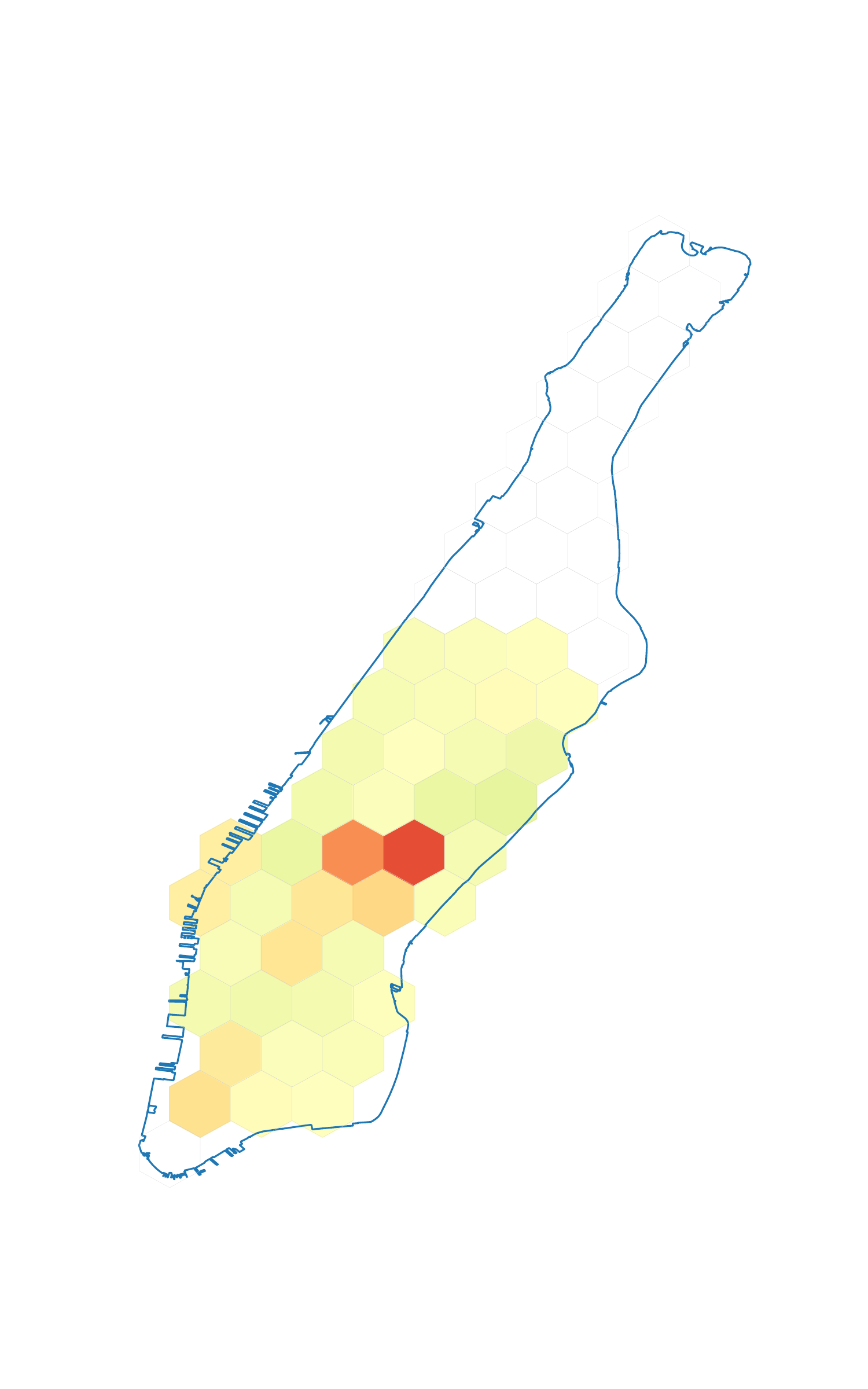}{}
        \caption{NYC - 38 zones}
        \label{fig:data_38n}
    \end{subfigure}
    \begin{subfigure}{.28\textwidth}
        \centering
        \includegraphics[width=\linewidth,trim={0 2.5cm 0 2.5cm},clip]{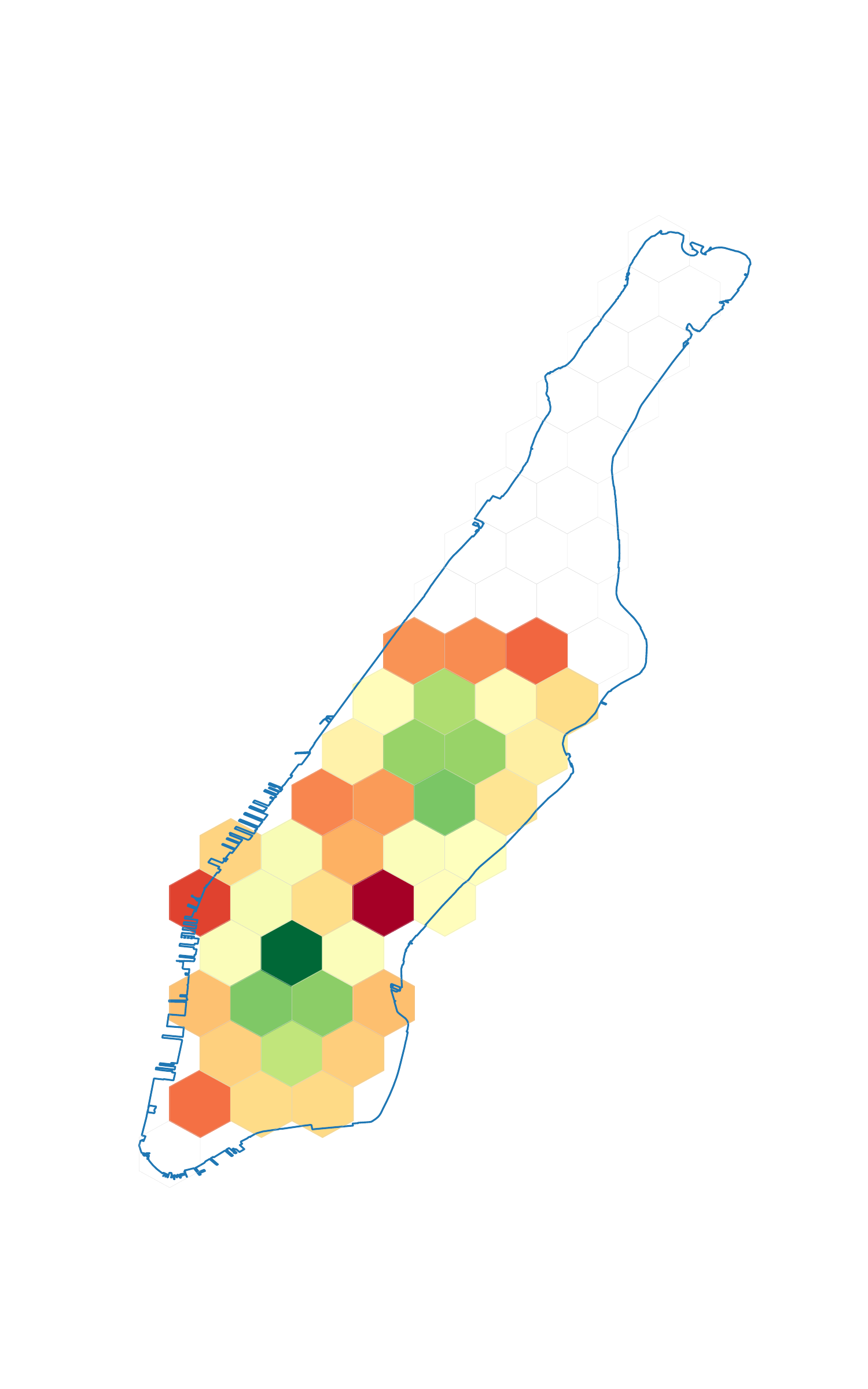}
        \caption{Clustered - 38 zones}
        \label{fig:data_38s}
    \end{subfigure}
\caption{The 38 zone operating area. Zones dominated by pick-ups are marked in green, zones dominated by drop-offs are marked in red, and zones with a balance between pick-ups and drop-offs are marked in light green/yellow.}
\label{fig:data_38}
\end{figure*}

\subsection{Benchmarks}\label{appendix_benchmarks}

In the following, we describe the used benchmarks for all experiments:

\textbf{\textit{Greedy}:} It chooses the highest reward at the current time step. Thus, it assigns every request a weight that depends on the profitability of the request. The weight is zero for unprofitable requests, i.e., if the cost is higher than the fare or if the time to pick-up the customer is higher than the maximum waiting time. The algorithm weighs all other requests proportionally to the profitability for each vehicle. To this end, \textit{Greedy} subtracts the cost of the total distance, i.e., the distance to the pick-up plus the distance driven with the customer, from the fare of the request. 

\textbf{\textit{Rebalancing Heuristic}:} It weighs customer requests the same way as the greedy policy, but has additional rebalancing capabilities. \textit{Heuristic} weighs rebalancing actions lower than customer requests and aims for an equal distribution of vehicles per zone. \textit{Heuristic} attributes a non-zero weight to rebalancing requests only if the vehicle's current zone has more vehicles than the rounded up average number of vehicles per zone $\lceil v_t \rceil$ and vehicle's destination zone has less vehicles than the rounded down $\lceil v_t \rceil$ average number of vehicles per zone. The weight is proportional to the vehicle's distance to the destination zone.

\textbf{\textit{RVC}:} Each agent represents a request-vehicle pair. All agents return two probabilities, one for declining the request and one for accepting it. Both probabilities add up to one. The accept probabilities serve as weights in the bipartite matching problem if they are higher than the reject probabilities during validation and testing. If the accept probabilities are lower than the reject probabilities, they are not considered in the bipartite matching problem. During training, accept and reject decisions are sampled according to their respective probabilities and the weights in the bipartite matching problem then correspond to the probability of the sampled action.

\subsection{Rebalancing Behavior} \label{sec:appendix_rebalancing}

Our experiments show for the 11 small zones setting that $V_{\mathrm{ext.}}$ always rebalances vehicles to a neighboring zone in one time step. This behavior is beneficial as the agent can reevaluate its rebalancing actions in each time step and check for new customer requests. The larger the setting, the more rebalancing actions the agent has to evaluate. This increases noise and computational effort. As the agent exclusively rebalances to neighboring zones in the context of the 11 small zones setting, we intentionally extend this practice to encompass only neighboring zones as viable rebalancing options for the larger 38 zones.

\end{document}